\documentclass{article}

\usepackage[english]{babel}

\usepackage[a4paper,top=2cm,bottom=2cm,left=3cm,right=3cm,marginparwidth=1.75cm]{geometry}

\usepackage{amsmath}
\usepackage{graphicx}
\usepackage[colorlinks=true, allcolors=blue]{hyperref}
\usepackage{amsthm}

\newtheorem{theorem}{Theorem}

\newcommand{\cent}[0]{\mbox{\textcent}}
\newcommand{\dollar}[0]{\$}

\newcommand{\mymatrix}[2]{\left( \begin{array}{#1} #2\end{array} \right)}

\newcommand{\myvector}[1]{\mymatrix{c}{#1}}

\newcommand{\tildesigma}{\widetilde{\Sigma}}

\newcommand{\powereq}{\mathtt{PowerEQ}}

\title{Real-valued affine automata compute beyond Turing machines\thanks{Part of this work was done while Yakary{\i}lmaz was visiting Institute of Theoretical and Applied Informatics, Gliwice, Poland in Autumn 2022.}}
\author{Abuzer Yakary\i lmaz$^{1,2}$ \\ ~ \\
$^1$Center for Quantum Computer Science, University of Latvia, R\={\i}ga, Latvia \and $^2$QWorld Association, Tallinn, Estonia, \url{https://qworld.net}
\\ ~ \\
\textit{email: abuzer.yakaryilmaz@lu.lv}
} 

\begin{document}
\maketitle

\begin{abstract}
We show that bounded-error affine finite automata (AfAs) recognize uncountably many 
(and so some non-Turing recognizable) languages when using real-valued transitions.
\end{abstract}

\section{Introduction}

Probabilistic finite automata (PFAs) have non-negative transition values called probabilities \cite{Rab63}. Quantum finite automata (QFAs), on the other hand, have complex-valued transition values called amplitudes \cite{AY21}. Here each complex-valued transition can be replaced by real-valued transitions by doubling the set of states. Thus, what makes QFAs different from PFAs is using both negative and positive transition values, which allows us to interfere the transitions constructively and destructively. Indeed, QFAs are arguably the simplest computational models to investigate the power of interference (see e.g., \cite{af98,YS10A,AmbY12,SY14,ShurY16}).

Interestingly, the idea of using both negative and positive valued transitions for classical automata already exists in the literature \cite{Hir18}, as old as PFA itself, called generalized automata (GAs) \cite{Tur68,Tur69}. However, up to our knowledge, GAs have never been considered as a model to define languages with bounded error. More specifically, both PFAs and QFAs assign a probability value (between 0 and 1) to each input, but, GAs assign a real value to each input (see \cite{Paz71} for the details about GAs). 

Affine finite automaton (AfA) is a quantum-like generalization of probabilistic finite automaton (PFA) \cite{DCY16}. AfAs can have both negative and positive transition values, with the restriction that, similar to PFAs, the entry summation of state vector must be 1. Then, the operators of AfAs must preserve the summation of state vectors. Similar to QFAs, AfAs can use interference as a computational resource, and, they are a generalization of PFAs as a PFA is a restricted AfA using only with non-negative transition values. 

Another quantum-like property of AfAs is the weighting operator (similar to quantum measurement operators defined with $\ell_2$-norm): The final state vector is normalized with respect to  $\ell_1$-norm, and so each state is observed with some probabilities by guaranteeing that the total probability is 1. Even though the computation of AfAs evolves linearly, the weighting operator does not necessarily to be linear.\footnote{In the abbreviation of AfAs, we use lower case ``f'' to emphasis this non-linearity.}

AfAs have been investigated in a series of papers and compared with QFAs and PFAs. In bounded-error and unbounded-error settings, AfAs are more powerful than both QFAs and AfAs \cite{DCY16,Yak22A}.\footnote{See also \cite{IKPY21,NKPVY22} for affine OBDDs and AfAs with a counter.} Bounded-error AfAs can also be more succinct than both bounded-error QFAs and PFAs \cite{BMY17,VY18,IKPY21,Yak21A}. In case of nondeterministic acceptance mode, AfAs are equaivalent to QFAs, and more powerful than PFAs \cite{DCY16}.

If we use only rational-valued transitions, then bounded-error AfAs can be simulated in deterministic logarithmic space \cite{HMY17,HMY21}. If we use real-valued transitions, then bounded-error AfAs can verify every unary languages similar to bounded-error two-way QFAs (2QCFAs) \cite{KY21A}. 

In this paper, we show that bounded-error AfAs recognize uncountably many 
(and so some non-Turing recognizable) languages when using real-valued transitions. A similar result was shown for bounded-error 2QCFAs \cite{SY17}. Here, we re-construct the proof given for 2QCFAs for AfAs. Both algorithms make consecutive equality check of several blocks of symbols. The 2QCFA algorithm use the polynomial-time algorithm given in \cite{AW02} for each pair of blocks, where the blocks are read many times. On the other hand, AfAs read the input once, and so, we make consecutive equality checks in real time.  Both methods use the encoding technique presented in \cite{ADH97} about the membership of any given language.

We give the definitions and notations used throughout the paper in the next section. In Section~\ref{sec:linear}, we present certain linear operations, which are used in our proofs. We present the AfA algorithm for successive equality check in Section~\ref{sec:equal}. We present the main result in Section~\ref{sec:uncountable}.

\section{Preliminaries}
\label{sec:pre}

We assume the reader is familiar with the basics of automata theory (see e.g., \cite{Sip13}).
We denote the input alphabet $\Sigma$ not containing the left and right end-markers ($\cent$ and $\dollar$, respectively); and $\tildesigma = \Sigma \cup \{\cent,\dollar\}$. We denote the empty string $\varepsilon$ (its length is zero) and it is fixed for every alphabet. For any given non-empty string $x \in \Sigma^*$, $x_i$ represents its $i$-th symbol from the left.

For a given $m>0$, $\{e_1,e_2,\ldots,e_m\}$ represent the set of standard basis (column) vectors in $\mathbf{R}^m$. For a given (column) vector $v$ in $\mathbf{R}^m$, $v[j]$ is its $j$-th entry and $|v|$ is the $\ell_1$-norm of $v$:
\[
    |v| = |v[1]|+\cdots+|v[m]|.
\]

An $m$-state AfA $M$ is a 5-tuple $M=(\Sigma,S,\{A_\sigma \mid \sigma \in \tildesigma\},s_I,S_a)$, where
\begin{itemize}
    \item $S = \{s_1,\ldots,s_m\}$ are the set of states,
    \item $A_\sigma$ is the affine operator representing the transitions when reading symbol $\sigma \in \tildesigma$,
    \item $s_I \in S$ is the initial state, and,
    \item $S_a \subseteq S$ is the set of accepting state(s).
\end{itemize}

The computation of $M$ is traced by an $m$-dimensional column vector called state vector. Its $j$-th entry represents the value of $M$ being in the state $s_j$. The entry summation of state vector must be 1. The computation starts in the basis vector $v_0 = e_I$ (as $s_I$ is the initial state). Any given input $x \in \Sigma^*$ with length $l$, $M$ reads the input from left to right and symbol by symbol including end-markers as
\[
    \cent~x_1~x_2~\ldots~x_l~\dollar. 
\]
Then, the computation evolves as
\[
v_f = A_{\dollar} A_{x_l} A_{x_{l-1}} \cdots A_{x_1} A_{\cent} e_I, 
\]
where $v_f$ is the final affine state vector. At the end, the weighting operator is applied and the state $s_j$ is observed with probability
\[
    \frac{|v_f[j]|}{|v_f|}.
\]
Thus, the input $x$ is accepted with probability
\[
    f_M(x) = \dfrac{\sum_{s_j \in S_a}|v_f[j]|}{|v_f|},
\]
and so it is rejected with probability $1 - f_M(x)$.

A language $L \subseteq \Sigma^*$ is said to be recognized by $M$ with error bound $\epsilon < \frac{1}{2}$ if and only if
\begin{itemize}
    \item every $x \in L$ is accepted by $M$ with probability at least $1 - \epsilon$ and
    \item every $x \notin L$ is rejected by $M$ with probability at least $1 - \epsilon$.
\end{itemize}
Then, it is also said $L$ is recognized by $M$ with bounded error or bounded-error $M$ recognizes $L$.

\section{Linear operations}
\label{sec:linear}

In our algorithms, we use linear operations (matrices) to make calculations with the entries of the state vector while reading the input.
\begin{itemize}
    \item We increase/decrease the value of an entry by some integers $d$: 
    \[ \mymatrix{cc}{1 & 0 \\ d & 1 } \myvector{1 \\ a } = \myvector{1 \\ a + d}. \]
    \item Then, we can count the number of symbols:
    \[
        \mymatrix{cc}{1 & 0 \\ 1 & 1 }^m \myvector{1 \\ 0} = \myvector{1 \\ m},
    \]
    where we suppose to read $m$ symbols and so apply the counting operator $m$ times. For each symbol, we may count by $k \in \mathbf{Z}$:
    \[
        \mymatrix{cc}{1 & 0 \\ k & 1 }^m \myvector{1 \\ 0} = \myvector{1 \\ km},
    \]
    \item We subtract one entry from another:
    \[  \mymatrix{cc}{1  & -1 \\ 0 & 1 } \myvector{a \\ b } = \myvector{a -b \\ b}. \]
    \item We multiply one entry with a value $c$:
    \[  \mymatrix{cc}{1  & 0 \\ 0 & c } \myvector{a \\ b } = \myvector{a \\ b c }. \]
    \item We make counter-clockwise rotations with angle $\theta$ on the unit circle:
    \[
    \mymatrix{cc}{ \cos \theta  & -\sin \theta \\ \sin \theta & \cos \theta } \myvector{ 1 \\ 0 } = \myvector{ \cos \theta \\ \sin \theta}
    \mbox{ and }
    \mymatrix{cc}{ \cos \theta  & -\sin \theta \\ \sin \theta & \cos \theta }^m \myvector{ 1 \\ 0 } = \myvector{ \cos m \theta \\ \sin m \theta}.
    \]
    Remark that any point on the unit circle can be represented by a single angle $ \theta_1 $, and then the point is $(\cos \theta_1, \sin \theta_1)$. If we apply a rotation with angle $\theta_2$ $k$ times, then the obtained point has the angle $ \theta_1 + k \theta_2 \mod 2\pi $.
\end{itemize}

We remind the following basic rules about linear operators, where $A$ and $B$ are matrices and $u$ and $v$ are vectors:
\begin{itemize}
    \item $(A \otimes B) \cdot (u \otimes v) = (A \cdot u) \otimes (B \cdot v)$
    \item $(A \otimes B)^m \cdot (u \otimes v) = (A^m \cdot u) \otimes (B^m \cdot v)$
    \item $\mymatrix{c|c}{ A & 0 \\ \hline 0 & B } \myvector{u \\ \hline v} = \myvector{ A\cdot u \\ \hline B \cdot v } $
\end{itemize}

For a given linear computation
\[
    v' = A_m A_{m-1} \cdots A_1 v,
\]
we use the following simulation by AfAs
\[
\myvector{\\ v' \\ \hline \\ \overline{1}} = 
\mymatrix{c|c}{ \\ A_m & 0 \\ \\ \hline \\ \overline{1} \cdots \overline{1} & 1 }
\mymatrix{c|c}{ \\ A_{m-1} & 0 \\ \\ \hline \\ \overline{1} \cdots \overline{1} & 1 } \cdots 
\mymatrix{c|c}{ \\ A_1 & 0 \\ \\ \hline \\ \overline{1} \cdots \overline{1} & 1 }
\myvector{\\ v \\ \hline \\ \overline{1}},
\]
where $\overline{1}$s are the values to make each column summation 1 so that each operator and vector become affine. 

\section{Consecutive equality check in real-time}
\label{sec:equal}

In this section, we present a bounded-error AfA algorithm for the language
\[
    \powereq = \{ a^{7\cdot 8^0} b a^{7\cdot 8^1} b a^{7\cdot 8^2} b a^{7\cdot 8^3} b \cdots b a^{7\cdot 8^n} \mid n \geq 0 \},
\]
which is slightly different than the language given in \cite{SY17}. Remark that, for each $n\geq 0$, $\powereq$ has one member string that has $8^{n+1}-1$ $a$ symbols.

\begin{theorem}
    The language $\powereq$ is recognized by an AfA $M$ with bounded error.
\end{theorem}
\begin{proof}
    We construct our AfA $M$ step by step. It has 19 states $\{s_1,\ldots,s_{19}\}$. The initial state is $s_1$, which is also only the accepting state.
    
    Let $x$ be a given input with $n \geq 0$ $b$ symbol(s):
    \[
    a^{t_0} b a^{t_1} b a ^{t_2} b a^{t_3} b \cdots b a^{t_n}, 
    \]
    where $t_0,t_1,\ldots,t_n \geq 0$. It is easy to see that $ x $ is in $\powereq$ if and only if 
    \begin{itemize}
        \item $t_0 = 7$ and
        \item $t_{j} = 8 t_{j-1} $ for $j  = 1,\ldots,n$.
    \end{itemize}  

    The AfA $ M $ encodes $ t_j $ or $8t_j$ into the values of some states while reading $t_j$ $a$s. Besides it makes the subtraction $ (t_{j+1} - 8 t_j )  $, and, by tensoring, it calculates $ (t_{j+1} - 8 t_j )^2 $ which is always a non-negative integer. Thus, in order to determine the membership of $x$, $M$ calculates the following cumulative summation
    \[
        T_x = (t_0-7)^2 + (t_1 - 8 t_0)^2 + (t_2 - 8 t_1)^2 + (t_3 - 8 t_2)^2 + \cdots + (t_n-8 t_{n-1})^2
    \]
    as the value of some states. Observe that if $x$ is in $\powereq$, then the summation $T_x$ is equal to 0, and, if $x$ is not in $\powereq$, then the summation $T_x$ is a positive integer. 

    The final state vector of $M$ on $x$ is
    \begin{equation}
        \label{eq:vf}
        v_f = \myvector{1 \\ k \cdot T_x \\ - k \cdot T_x \\ 0 \\ \vdots \\ 0  }
    \end{equation}
    for some integer $k>1$. That is, if $x \in \powereq$, then it is accepted with probability 1 as the second and third entries will be zero. If $x \notin \powereq $, then it is rejected with  probability $\frac{2k}{2k+1}$, which can be arbitrarily close to 1 by picking large $k$ values.

    The computations starts in $v_0 = \left( 1 ~~ 0 ~~ \cdots ~~ 0  \right)^T$. During the computation, the state vector is split into four parts as
    \[
        v = 
     \myvector{ v^{'} \otimes v^{'} \\ \\ \hline \\ v^{''} \otimes v^{''} \\ \\ \hline \\ T \\ \\ \hline \\ \overline{1} },
    \]
    where
    \begin{itemize}
        \item both $v^{'}$ and $v^{''}$ are integer vectors in the form of $\myvector{1 \\ a \\ b}$ for some $a,b \in \mathbf{Z}$, 
        \item $T$ is the partial summation of $T_x$ until that moment of the computation, and,
        \item $\overline{1}$ is a value to guarantee that the summation of all entries is equal to 1.
    \end{itemize}

    Suppose that we are reading the block $a^{t_j}$ for $ j \in \{0,\ldots,n\} $. The value of $t_j$ is encoded on the third state of $v'$ and the value of $ 8 \cdot t_j $ is encoded on the second state of $v''$. Here the second entry of $v'$ is already set to $8 \cdot t_{j-1} $ (except $j=0$, see below) and the third entry of $v''$ is 0. Thus, after reading $a^{t_j}$, we have
    \[
        v' = \myvector{1 \\ 8  t_{j-1} \\ t_j } 
        \mbox{ and }
        v'' = \myvector{1 \\ 8 t_j \\ 0}.
    \]
    The case for $j=0$ is handled by using the left-end marker: After reading $\cent$, we set
    \[
        v' = \myvector{1 \\ 7 \\ 0}
        \mbox{ , }
        v'' = \myvector{1 \\ 0 \\ 0}
        \mbox{ , and }
        T = 0.
    \]

    When reading a $b$ after the block $a^{t_j}$, we do the following linear operations:
    \begin{enumerate}
        \item The vector $v'$ is updated as $\myvector{1 \\ t_j - 8 t_{j-1} \\ 0 } $. 
        \item The fifth entry of $v' \otimes v'$ becomes $ (t_j - 8 t_{j-1} ) ^ 2 $. This value is added to $T$.
        \item The vector $v' \otimes v'$ is set to $ 
    v'' \otimes v''= \myvector{1 \\ 8t_j \\ 0 } \otimes \myvector{1 \\ 8t_j \\ 0 }  $.
        \item The vector $v''$ is set to $\myvector{1 \\ 0 \\ 0}$.
    \end{enumerate}
    For each item, there is an integer-valued linear matrix, where the column summation does not necessarily to be 1. By using the last row, we can easily make them affine operator as explained in the previous section.  So, the affine operator for $b$ is the multiplication of all of them.

    When reading $\dollar$, we do the operations Items 1 and 2 given for reading symbol $b$, and then we put the final vector in the form of $v_f$, see Eq.~\ref{eq:vf}.
\end{proof}

\section{Computing uncountable many languages}
\label{sec:uncountable}

For any given finite alphabet $\Sigma$, we enumerate all possible strings in lexicographic order. For example, if $\Sigma = \{a\}$ or $\Sigma=\{a,b\}$, then we have 
\[
\begin{array}{lll}
     0 & \varepsilon & \varepsilon \\
     1 & a & a \\
     2 & aa & b \\
     3 & a^3 & aa \\
     4 & a^4 & ab \\
     5 & a^5 & ba \\
     6 & a^6 & bb \\
     7 & a^7 & aaa \\
     \vdots & \vdots & \vdots \\
\end{array}.
\]
The following idea is based on the enumeration of the strings, and so, it is applicable to any finite alphabet. For simplicity, we focus on unary alphabet $\Sigma = \{a\}$. Remark that all unary languages form an uncountable set.

Let $ L $ be a language defined on $\{a\}$. We use the encoding technique presented in \cite{ABGKMT06}. Here is a function for the membership of strings as
\[
F_L\left( a^j \right) = \left\{ \begin{array}{rl}
   1,  &  \mbox{if } a^j \in L \\
   -1,  & \mbox{if } a^j \notin L
\end{array} \right. ,
\] 
and here is an angle (in radian) to store all membership information:
\[
\theta_L = 2 \pi \sum_{i=0}^{\infty} \left( \frac{F_L(a^i)}{8^{i+2}}\right).
\]

We define language $\powereq(L)$ as
\[
    \powereq(L) = \{ a^{7\cdot 8^0} b a^{7\cdot 8^1} b a^{7\cdot 8^2} b a^{7\cdot 8^3} b \cdots b a^{7\cdot 8^n} \mid n \geq 0 \mbox{ and } a^n \in L \}.
\]
Similar to all unary languages, the set $\left\{ \powereq(L) \mid L \subseteq \{a\}^* \right\}$ is uncountable and so contains non-Turing recognazible languages.

\begin{theorem}
    For a given unary language $L$, $\powereq(L)$ is recognized by an AfA $M_L$ with bounded error.
\end{theorem}
\begin{proof}
    We present a 5-state AfA $M_L'$, and then tensor it with the AfA $M$ given in the previous section. We start with the idea of how to use angle $\theta_L$.
    
    We use the rotation with angle $\theta_L$ on the unit circle (on $\mathbf{R}^2$). We start with angle 0, i.e., the vector is $\myvector{1 \\ 0}$. For $ j = 0,1,2,\ldots$, if we rotate with angle $\theta_L$ $8^{j+1}$ times, we obtain the angle
    \[
        8^{j+1} \cdot 2 \pi \sum_{i=0}^{\infty} \left( \frac{F_L(a^i)}{8^{i+2}}\right) = 
        \pi \frac{F_L(a^j)}{4} + 2\pi \left( \frac{F_L(a^{j+1})}{8^2} + \frac{F_L(a^{j+2})}{8^3} + \cdots \right) .
    \]
    After that, if we make another rotation with angle $\frac{\pi}{4}$, we obtain the angle
    \[
        \phi_j = 
        \frac{\pi}{4} \left( F_L(a^j) + 1 \right) + 2\pi \left( \frac{F_L(a^{j+1})}{8^2} + \frac{F_L(a^{j+2})}{8^3} + \cdots \right) .
    \]
    Observe that if $a^j \in L$, then this angle is close to $\frac{\pi}{2}$, and, if $a^j \notin L$, then it is close to $0$. Let the following vector be the corresponding point on the unit circle:
    \[
        \myvector{\cos \phi_j \\ \sin \phi_j}.
    \]
    As given in \cite{ABGKMT06}, if $a^j \in L$, then $\sin^2 \phi_j$ is not less than $0.98$, and, if $a^j \notin L$, then $\cos^2 \phi_j$ is not less than $0.98$.

    Now, we present the details of the 5-state AfA $M_L'$. Its state vector is split into two parts as
    \[
        \myvector{ u \otimes u \\ \\ \hline \\ \overline{1} }.
    \]
    The initial state vector is
    \[
        \myvector{ \myvector{1 \\ 0 } \otimes \myvector{1 \\ 0} \\ \\  \hline \\ 0 }.
    \]
    Whenever symbol $ a $ is read, the rotation with angle $\theta_L$ is applied to $u$. Whenever symbol $b$ is read, the identity operator is applied. When reading $\dollar$, first the vector $u$ is rotated with angle $\frac{\pi}{4}$:
    \[
        \myvector{ \myvector{\cos \alpha \\ \sin \alpha} \otimes \myvector{\cos \alpha \\ \sin \alpha} \\ \\ \hline \\ \overline{1}  }  = \myvector{ \cos^2 \alpha \\ \cos \alpha \cdot \sin \alpha \\ \sin\alpha \cdot \cos \alpha \\ \sin^2 \alpha \\ -2\sin \alpha \cdot \cos \alpha },
    \]
    where $\alpha$ is the final angle. Second, we apply the following affine operator
    \[
        \mymatrix{rrrrr}{0 & 0 & 0 & 1 & 0 \\ 1 & 0 & 0 & 0 & 0 \\ 0 & 1 & 1 & 0 & 1 \\ 0 & 1 & 1 & 0 & 1 \\ 0 & -1 & -1 & 0 & -1}
    \]
    and so we obtain the final state vector as
    \[
        \myvector{ \sin^2 \alpha \\ \cos^2 \alpha \\ 0 \\ 0 \\0 }.
    \]

    The AfA $M_L$ runs both AfAs $M$ and $M_L'$ on the given input in parallel (their state vector and operators (matrices) are tensored). The final state vector of $M_L$ is set to 
    \[
        v_f'' = \myvector{ \sin^2 \alpha \\ \cos^2 \alpha \\ k \cdot T_x ( \sin^2 \alpha + \cos^2 \alpha) \\ - k \cdot T_x ( \sin^2 \alpha + \cos^2 \alpha) \\ 0 \\ \vdots \\ 0 }
    \]
    
    The automaton $M_L$ has 95 states and the first state is the only accepting state, which is also the initial state. Let $x$ be the given input. 

    \textbf{The case $x \in \powereq(L)$:} As $x$ is in $\powereq$, $T_x = 0$. So, the decision is made solely on the values $\sin^2 \alpha$ and $\cos^2 \alpha$. As described above, $\sin^2 \alpha$ is not less than 0.98 and so $x$ is accepted with probability no less than 0.98. 

    \textbf{The case $x \notin \powereq(L)$:} If $x \notin \powereq$, then $T_x \geq 1$ and so it is rejected with probability at least $\frac{2k + \cos^2 \alpha}{ 2k +1 }$, which is no less than $\frac{2k}{2k+1}$. If $x \in \powereq$, then $T_x = 0$ and so the decision is made solely on the values $\sin^2 \alpha$ and $\cos^2 \alpha$. As described above, $\cos^2 \alpha$ is not less than 0.98 and so $x$ is rejected with probability no less than 0.98. 
\end{proof}

\section*{Acknowledgment}

Yakary{\i}lmaz was partially supported by the ERDF project Nr. 1.1.1.5/19/A/005 ``Quantum computers with constant memory'' and the project ``Quantum algorithms: from complexity theory to experiment'' funded under ERDF programme 1.1.1.5.

\bibliographystyle{plain}
\bibliography{ref}

\end{document}